\documentclass[submission,copyright,creativecommons]{eptcs} 
\providecommand{\event}{EXPRESS/SOS 2022} 
\providecommand{\thisvolume}[1]{this volume of EPTCS, Open Publishing Association} 

\usepackage{amssymb}
\usepackage{calculi}
\usepackage{amsthm}
\usepackage[textsize = scriptsize, disable]{todonotes}
\def\note#1{\todo[color = green]{#1}}

\theoremstyle{plain}
\newtheorem{definition}{Definition} 
\newtheorem{theorem}{Theorem}
\newtheorem{lemma}{Lemma}
\newtheorem{proposition}{Proposition}
\newtheorem{corollary}{Corollary}

\def\INCRSYM{\ensuremath{\text{\syntaxfont{+}}}}
\def\RINCR#1{\ensuremath{#1\INCRSYM}}
\def\LINCR#1{\ensuremath{\INCRSYM#1}}
\def\NCOMP#1#2{\ensuremath{#1\cdot#2}}
\def\MCOMP#1{\NCOMP{#1}{#1}}
\def\DCOPY#1{\ensuremath{D\ARGS{#1}}}
\def\OUT#1{\ensuremath{\overline{#1}}}
\def\ssubst{\sigma_{\kern-2pt s}}
\def\tsubst{\sigma_{\kern-2pt t}}
\def\QDEPTH#1{{\normalfont\textsc{qd}}\ARGS{#1}}
\def\NAMESPACE#1{\SETNAME{N}_{\kern-3pt\HOLE[#1]}}
\newcommand{\NAMESERVER}[1][s]{\ensuremath{\REPL{N(x,z,v,#1)}}}
\def\LNS#1{\prescript{\INCRSYM\kern-1pt}{}{\NAMESPACE{#1}}}
\def\RNS#1{\NAMESPACE{#1}^{\INCRSYM}}
\def\CNS#1{\NAMESPACE{#1}^{\circ}}
\def\LNC{\prescript{\INCRSYM\kern-1pt}{}{N}}
\def\RNC{N^{\INCRSYM}}
\def\CNC{N^{\circ}}

\newcommand{\RHOBISIM}[1][\FN{P}]{\ensuremath{\BISIM^{#1}}}
\newcommand{\RHOWBISIM}[1][\FN{P}]{\ensuremath{\WBISIM^{#1}}}

\ifengineTF{\pdftexengine}{
  \def\ssubst{\sigma_s}
  \def\tsubst{\sigma_t}
  \def\NAMESPACE#1{\SETNAME{N}_{\kern-1pt\HOLE[#1]}}
  \def\LNS#1{\prescript{\INCRSYM\kern-6pt}{}{\NAMESPACE{#1}}}
  \def\NAMEEQ{\ensuremath{\equiv_{\kern-4pt\mathcal{N}}}}
}{}

\title{Encodability and Separation \\ for a Reflective Higher-Order Calculus%
\thanks{This work was supported by the Icelandic Research Fund Grant No.\@ 218202-05(1-3).}}
\author{Stian Lybech
\institute{Reykjavík University\\
Reykjavík, Iceland}
\email{stian21@ru.is}
}
\def\titlerunning{Encodability and Separation for a Reflective Higher-Order Calculus}
\def\authorrunning{S. Lybech}
\begin{document}
\maketitle

\begin{abstract}
The \RHO-calculus (Reflective Higher-Order Calculus) of Meredith and Radestock is a \PI-calculus-like language with some unusual features, notably, structured names, runtime generation of free names, and the lack of an operator for scoping visibility of names.
These features pose some interesting difficulties for proofs of encodability and separation results.
We describe two errors in a previously published attempt to encode the \PI-calculus in the \RHO-calculus by Meredith and Radestock.
Then we give a new encoding and prove its correctness, using a set of encodability criteria close to those proposed by Gorla, and discuss the adaptations necessary to work with a calculus with runtime generation of structured names.
Lastly we prove a separation result, showing that the \RHO-calculus cannot be encoded in the \PI-calculus.
\end{abstract}

\section{Introduction}
Process calculi are formalisms for modelling and reasoning about concurrent and distributed computations; a prominent example is the \PI-calculus of Milner, Parrow and Walker \cite{milner_walker_parrow1992picalc, PICALC}.
These languages commonly begin by assuming a \emph{countably infinite set of atomic names} $\NAMES$, ranged over by $x,y,z$.
This is not an unreasonable assumption for most purposes, but it does leave open the question of how this set of names should actually be interpreted, e.g.\@ if we were to create an implementation of the \PI-calculus or one of its variants \cite{turner1996phd, pierce2000pict, fournetgonthier2000joincalculus}.


A similar issue arises with the scoping operator $\NEW{x}P$, which is used to declare a new name $x$ with visibility limited to $P$.
Here the question becomes how we should choose this new name $x$, such that it is actually ensured to be unique.
For a process modelling a program running on a single computer, this can easily be solved, e.g.\@ with a counter; but if the process models a \emph{distributed} system, with programs running on distinct computers, the solution is less obvious. 
These issues are not directly handled in the \PI-calculus model, but only become apparent when we consider a more practical implementation of the set of names.

A radically different approach is taken in the Reflective Higher-Order (RHO or \RHO) calculus proposed by Meredith and Radestock in \cite{RHOCALC}.
These authors instead begin by positing that the set of names is built by a syntax, similar to the syntax for processes, and thus generated from a \emph{finite} set of elements.
One could imagine different possibilities for this syntax, but Meredith and Radestock here make the unusual choice of letting names be `quoted' processes, written $\QUOTE{P}$.
Thus, if $P$ is a process, then $\QUOTE{P}$ is a name.
This creates a mutually recursive definition, since processes also contain names.
The full syntax of the \RHO-calculus is then
\begin{center}
\begin{syntax}[h]
  P \in \PROC[rho] \IS \NIL 
    \OR P_1 \PAR P_2
    \OR \LIFT{x}{P}
    \OR \INPUT{x}{y} . P
    \OR \DROP{x}
\tabularnewline
x, y \in \QUOTE{\PROC[rho]} \IS \QUOTE{P}
\end{syntax}
\end{center}

Three of the constructs are as in the \PI-calculus: The \emph{nil} process, $\NIL$, is the inactive process; The \emph{parallel} construct, $P_1 \PAR P_2$, is the parallel composition of processes $P_1$ and $P_2$; and the \emph{input} construct, $\INPUT{x}{y}.P$, is a blocking operation, awaiting a communication on the channel $x$ of some name, which, upon reception, will be bound to $y$ in the continuation $P$.

The two remaining constructs are particular to the \RHO-calculus:
The \emph{lift} construct $\LIFT{x}{P}$ quotes the process $P$, thereby creating the name $\QUOTE{P}$, and outputs it on $x$; thus name generation is handled explicitly in the \RHO-calculus, rather than implicitly by a \PI-calculus style $\nu$-operator.
This is the second peculiarity of this calculus, since the newly generated name will be \emph{free} in the continuation of the corresponding input, and therefore also \emph{observable} if substituted for the subject of an input or lift.
As we shall later see, this feature is crucial for showing a separation result w.r.t.\@ the \PI-calculus.

Lastly, the \emph{drop} construct $\DROP{x}$ removes the quotes of the name to run the process within them, thereby enabling higher-order behaviour (i.e.\@ process mobility). 
This construct is thus similar to a process variable $X$ in e.g.\@ \HOPI{} \cite{SANGIORGIPHD, HOPICALC}, and is also the reason for the `reflective' epithet in the name of this calculus. 
It derives from Smith \cite{smith1982procedural_reflection}, who defined reflection as the ability of a program to turn code into data, compute with it, modify it, and turn it back into running code, which in the \RHO-calculus is captured by the combination of the lift and drop constructs, and the duality of names and processes.

Although superficially quite similar to the \PI-calculus, these features suggest that the \RHO-calculus is actually rather different.
As argued above, the use of structured terms as names, and explicit name generation, seem more realistic from an implementation perspective, as it places the problems of choosing the next name, and of ensuring freshness, within the language itself, rather than simply assuming that these features just work behind the scenes.
However, providing a \emph{solution} to these problems is not trivial, as we shall see below. 
For example, in \cite{RHOCALC} Meredith and Radestock also propose an encoding of the asynchronous, choice-free fragment of the \PI-calculus into the \RHO-calculus, reviewed in section~\ref{sec:mrencoding}, but as we shall show in section~\ref{sec:errors}, this encoding contains two fatal errors, invalidating their correctness result.

In what follows, we shall instead propose a different encoding of the \PI-calculus into the \RHO-calculus and formally prove its correctness w.r.t.\@ a number of encodability criteria closely related to those proposed by Gorla in \cite{GORLA}, but with some adaptations necessitated by the aforementioned peculiar features of the \RHO-calculus (Propositions~\ref{prop:independence}-\ref{prop:divergence_reflection}).
Using the same criteria we then derive a separation result, showing that the converse of this statement does not hold: there cannot be an encoding of the \RHO-calculus into the \PI-calculus satisfying the same criteria (Theorem~\ref{thm:separation}).\footnote{%
Full proofs of most  results are available in a technical report \cite{rhocalc2022techreport}.}
This result is quite surprising, and it suggests that we cannot always just reduce higher-order behaviour to the first-order paradigm, as Sangiorgi was able to do with \HOPI{} in \cite{HOPICALC}.
This is because higher-order behaviour in the \RHO-calculus is not just an extension on top of an already computationally complete language, as it is the case with \HOPI{} which extends the `first-order' \PI-calculus, but rather appears as a special case of the more general phenomenon of reflection, where processes (code) are communicated without modification.

\section{The Reflective Higher-Order Calculus}\note{I have completely reordered this section, because R1 found it confusing.}
We begin by presenting the \RHO-calculus following Meredith and Radestock in \cite{RHOCALC}.
As we have already seen the syntax above, we shall here focus on the semantics, which we shall give in terms of a reduction system.
Firstly, we shall need a notion of \emph{structural congruence} on processes, written $\STRUCTEQ$.
We shall postpone its precise definition slightly, but the intuition is that $P_1 \STRUCTEQ P_2$ denotes that $P_1$ and $P_2$ are the same process, up to some insignificant structural change, such as reordering of components in parallel composition or a change of bound names (\ALPHA-conversion).

Now, since names are quoted processes, this notion of structural congruence is extended to the set of names: the \emph{name equivalence} relation, written $\NAMEEQ$, is defined as the least equivalence on names closed under the following rules:
\begin{center}
\begin{semantics}
  \RULE[n-struct][n:struct]
    { P_1 \STRUCTEQ P_2 }
    { \QUOTE{P_1} \NAMEEQ \QUOTE{P_2} }
\end{semantics}
\begin{semantics}
  \RULE[n-drop][n:drop]
    { x_1 \NAMEEQ x_2 }
    {\QUOTE{\DROP{x_1}} \NAMEEQ x_2}
\end{semantics}
\end{center}

\noindent The point of \nameref{n:struct} is that if the processes within quotes have the same structure (up to structural congruence), then the quoted processes should also represent the same name.
Furthermore, by \nameref{n:drop}, we allow nested levels of quotes and drops to `cancel out.'

Next, we shall need the notions of free and bound names, $\FN{P}$ and $\BN{P}$, which are defined in the usual (syntactic) way, with input being the only formal binder in the language.
Thus $\BN{\INPUT{x}{y}.P} = \SET{y} \UNION \BN{P}$, and all other names are free.
We write $\NAMESOF{P} \DEFSYM \FN{P} \UNION \BN{P}$ for all the names in $P$, and we also write $\FRESH{x}P$ to mean that $x$ is \emph{fresh} for $P$.
However, with \emph{structured} names, it is no longer enough that $x \notin \NAMESOF{P}$; $x$ must also not be \emph{name equivalent} to any name in $P$.
Thus we say $\FRESH{x}P \DEFSYM \forall n \in \NAMESOF{P} \SUCHTHAT x \not\NAMEEQ n$.
Lastly, we write $P\REPLACE{x}{y}$ for the safe substitution of $x$ for $y$ within $P$.
However, given our considerations about $\NAMEEQ$ above, $P\REPLACE{x}{y}$ will not only replace $y$, but also any name that is \emph{name equivalent} to $y$.
Note also, in particular, that substitution does \emph{not} recur into processes under quotes.
Thus $\QUOTE{P}\REPLACE{x}{y} = \QUOTE{P}$ for all names $y$ where $y \not\NAMEEQ \QUOTE{P}$, and $\QUOTE{P}\REPLACE{x}{y} = x$ otherwise.

We shall now return to the definition of structural congruence: it is defined as the usual least congruence on processes, containing \ALPHA-equivalence and the abelian monoid rules for parallel composition with $\NIL$ as the unit element.
However, with structured terms as names, the congruence rules take on a slightly unusual form, since we now also need to compare names.
For example, to conclude $\INPUT{x_1}{y_1}.P_1 \STRUCTEQ \INPUT{x_2}{y_2}.P_2$ we would need the following rule in structural congruence:
\begin{center}
\begin{semantics}
  \RULE[s-in](\FRESH{z}P_1,P_2)
    { x_1 \NAMEEQ x_2 \AND P_1\REPLACE{z}{y_1} \STRUCTEQ P_2\REPLACE{z}{y_2} }
    { \INPUT{x_1}{y_1}.P_1 \STRUCTEQ \INPUT{x_2}{y_2}.P_2 }
\end{semantics} 
\end{center}

\noindent This yields another mutual recursion between structural congruence and name equivalence.

With these concepts in place, we can at last give the reduction rules for our semantics as follows:
\begin{center}
\begin{semantics}
  \RULE[\RHO-par][r:par]
    {P_1 \trans P_1'}
    {P_1 \PAR P_2 \trans P_1' \PAR P_2}
\end{semantics}
\begin{semantics}
  \RULE[\RHO-struct][r:struct]
    {P_1 \STRUCTEQ P_1' \AND P_1' \trans P_2' \AND P_2' \STRUCTEQ P_2}
    {P_1 \trans P_2}
\end{semantics}

\begin{semantics}
  \RULE[\RHO-com][r:com]
    { x_1 \NAMEEQ x_2 }
    { \INPUT{x_1}{y} . P_1 \PAR \LIFT{x_2}{P_2} \trans P_1\REPLACE{\QUOTE{P_2}}{y} }
\end{semantics}
\end{center}

The \nameref{r:par} and \nameref{r:struct} rules are standard (as in e.g.\@ the \PI-calculus); the former lets us conclude a reduction of one component in a parallel composition, whilst the latter allows us to rewrite the process, using structural congruence $\STRUCTEQ$, such that its form can match the conclusion of one of the other rules.

The \nameref{r:com} rule is also \emph{almost} standard:
The process $P_2$ is quoted and sent out over $x_2$, and the matching input receives it as the name $\QUOTE{P_2}$ and substitutes it for $y$ in the continuation $P_1$.
However, since names in the \RHO-calculus have structure, we must be able to explicitly conclude the equivalence $x_1 \NAMEEQ x_2$ between the two subjects in a communication.
This is thus different from calculi with atomic names where exact syntactic equality is (usually implicitly) required between subjects.

One last detail concerns substitution: 
In structural congruence, including \ALPHA-equivalence, $P\REPLACE{x}{y}$ is defined as the usual capture-avoiding substitution of names for names.
However, the substitution used in the \emph{semantics} is slightly different, as it is also used to handle the $\DROP{x}$ construct, which was not given a reduction rule above.
The semantic substitution also contains the clause $\DROP{x}\REPLACE{\QUOTE{P}}{y} = P$ if $x \NAMEEQ y$, thus replacing the \emph{process} $\DROP{x}$ with $P$; and $\DROP{x}\REPLACE{\QUOTE{P}}{y} = \DROP{x}$ if $x \not\NAMEEQ y$.
This is the only way in which a $\DROP{x}$ is ever executed, and it implies that the drop of a \emph{free} name is a deadlock, as it can never be touched by a substitution at runtime.

\section{The encoding of Meredith and Radestock}\label{sec:mrencoding}
In \cite{RHOCALC}, Meredith and Radestock proposed an encoding of the asynchronous, choice-free \PI-calculus, taking full abstraction w.r.t.\@ weak, barbed bisimilarity as their correctness criterion.
Unfortunately, that encoding is \emph{not} correct, as we shall now show. 
The counter-examples are instructive, as they highlight some of the difficulties inherent in working with a calculus without the assumption of an infinite set of atomic names and explicit scoping operators.

First, we recall the syntax and semantics of the asynchronous choice-free \PI-calculus, as given e.g.\@ in \cite{parrow2001introduction}.
Note that some of the constructs and concepts are similar to those found in the \RHO-calculus.
We shall therefore reuse some of the symbols and rely on context to distinguish whether a \PI-calculus or \RHO-calculus construct is meant.\note{Added these two lines to clarify overloading for R3.}
The syntax is:
\begin{formatcalculus}[pi]
\begin{center}
\begin{syntax}[h]
  P \in \PROC \IS \NIL \OR P_1 \PAR P_2 \OR \INPUT{x}{y}.P \OR \OUTPUT{x}{z} \OR \NEW{x}P \OR \REPL{P}
\end{syntax} 
\end{center}

\noindent The semantics is given in terms of a reduction system with the rules
\begin{center}
\begin{semantics}
  \RULE[\PI-com][pi:com]
    { }
    { \INPUT{x}{y}.P \PAR \OUTPUT{x}{z} \trans P\REPLACE{z}{y} }
\end{semantics}
\begin{semantics}
  \RULE[\PI-res][pi:res]
    { P \trans P' }
    { \NEW{x}P \trans \NEW{x}P' }
\end{semantics}
\end{center}

\noindent and with rules for parallel composition and structural congruence similar to those in the \RHO-calculus (rules \nameref{r:par} and \nameref{r:struct} above).
Structural congruence $\STRUCTEQ$ over $\PROC$ contains the same rules as in the \RHO-calculus, but with syntactic equality replacing name equivalence, and also the following rules for scoping and replication:
\begin{center}
\begin{math}
\begin{array}{r @{~} l}
  \NEW{x}\NIL     & \STRUCTEQ \NIL            \\
  \NEW{x}\NEW{y}P & \STRUCTEQ \NEW{y}\NEW{x}P
\end{array}\hspace{0.5cm}
\begin{array}{r @{~} l}
  \REPL{P}            & \STRUCTEQ P \PAR \REPL{P} \\
  \NEW{x}P_1 \PAR P_2 & \STRUCTEQ \NEW{x}\PAREN{P_1 \PAR P_2} \text{ if $x \notin \FN{P_2}$}
\end{array}
\end{math}
\end{center}
\end{formatcalculus}

Now for the encoding, assume a function $\PHI : \SETNAME{N} \to \QUOTE{\PROC[rho]}$ from \PI-calculus atomic names to \RHO-calculus names.
Since the set of \PI-calculus names is countably infinite, it can for example be mapped to the set of natural numbers.
The function \PHI{} could then be regarded as an enumeration of names (or a successor function), starting e.g.\@ from \QUOTE{\NIL} for the name $x_0$, and then letting the name $x_{i+1}$ be defined in terms of the name $x_i$ as for example $x_{i+1} \DEFSYM \QUOTE{\LIFT{x_i}{\NIL}}$.
In the sequel, we shall say that $\QUOTE{\LIFT{x}{\NIL}}$ is a \emph{left increment} of $x$, written $\LINCR{x}$.
Then we can generate a countably infinite sequence of names $x_0, x_1, x_2, \ldots$, starting from any name $x = x_0$, as $\LINCR{x} = x_1, \LINCR{\LINCR{x}} = \LINCR{x_1} = x_2, \ldots$ and so on.
This shows that the set of \PI-calculus names can be implemented as \RHO-names, as, by the definition of name equivalence and structural congruence, we have that $x \not\NAMEEQ \QUOTE{\LIFT{x}{\NIL}}$.

Correspondingly we can define $\RINCR{x} \DEFSYM \QUOTE{\INPUT{x}{\QUOTE{\NIL}}.\NIL}$ as a \emph{right increment} of $x$, which gives us another countably infinite sequence.
Another option is \emph{name composition} $\NCOMP{x}{y} \DEFSYM \QUOTE{ \LIFT{x}{\NIL} \PAR \INPUT{y}{\QUOTE{\NIL}}.\NIL}$, which yields yet another sequence with $x^2 = \MCOMP{x}, x^3 = \NCOMP{x^2}{x}, x^4 = \NCOMP{x^3}{x}, \ldots$ and so on.\note{Change the sequence of composition to address comment from R1.}

These are all examples of \emph{static quoting} techniques for consistent name generation, and each could be used to implement the function $\PHI$. 
Given such techniques, Meredith and Radestock then begin by assuming that all \PI-calculus names are already implemented as \RHO-names.
Their translation function $\PTRANS[P](n_0,p_0)$ requires two names as parameters, which must be chosen such that they are distinct from all the names in $P$, and furthermore that no name \emph{within} $P$ can ever be \emph{generated} from $n_0$ or $p_0$ by means of the aforementioned methods of static name generation. 
One way of ensuring this is by letting
\begin{equation*}
  n_0 = \QUOTE{ \prod_{x \in \FN{P} } \LIFT{x}{\NIL} } \qquad\text{and}\qquad %
  p_0 = \QUOTE{ \prod_{x \in \FN{P} } \INPUT{x}{\QUOTE{\NIL}}.\NIL }
\end{equation*}

\noindent where $\prod$ denotes generalised parallel composition.

The translation function also uses two short-hands: $\DCOPY{x} \DEFSYM \INPUT{x}{y} . \PAREN{\DROP{y} \PAR \LIFT{x}{\DROP{y}} }$ is a copying process used to implement replication; and $\OUTPUT{x}{y} \DEFSYM \LIFT{x}{\DROP{y}}$ simulates output in the \PI-calculus, since by \nameref{n:drop} we have that $\QUOTE{\DROP{y}} \NAMEEQ y$.
The translation $\PTRANS[P] = \PTRANS[P](n_0, p_0)$ \cite[p.\@ 13]{RHOCALC} is then given by the following recursive equations:\footnote{The translation has been adapted to use our notation for name increments, which we find more intuitive than $x^l$ and $x^r$, which is used in the original presentation. 
We also use $\OUTPUT{x}{y}$ rather than $x[y]$ for output, which is more in line with standard \PI-calculus notation.}
\begin{ptrans}{}{}{}{n,p}
\begin{center}
\begin{math}
\begin{array}{r @{~} l}
  \PTRANS[\NIL]                & = \NIL                                                                       \\
  \PTRANS[{\OUTPUT[pi]{x}{y}}] & = \OUTPUT{x}{y}                                                              \\
  \PTRANS[\INPUT{x}{y} . P]    & = \INPUT{x}{y} . \PTRANS[P]                                                   
\end{array}\hspace{1cm}
\begin{array}{r @{~} l}
  \PTRANS[P_1 \PAR P_2]        & = \PTRANS[P_1](\LINCR{n}, \LINCR{p}) \PAR \PTRANS[P_2](\RINCR{n}, \RINCR{p}) \\
  \PTRANS[\NEW{x}P]            & = \INPUT{p}{x} . \PTRANS[P](\LINCR{n}, \LINCR{p}) \PAR \OUTPUT{p}{n}         \\ 
                               & 
\end{array}
\end{math}
\begin{math}
\begin{array}{r @{~} l}
  \hspace{2.25cm}\PTRANS[\REPL{P}]            & = \LIFT{\NCOMP{n}{p}}{     \INPUT{\RINCR{n}}{n} . \INPUT{\RINCR{p}}{p} . \PAREN{ \PTRANS[P] \PAR \DCOPY{\NCOMP{n}{p}} \PAR \LIFT{\RINCR{n}}{ \OUTPUT{n}{n} } \PAR \LIFT{\RINCR{p}}{ \OUTPUT{p}{p} } } } \\ 
                               & ~\PAR \DCOPY{\NCOMP{n}{p}} \PAR \OUTPUT{\RINCR{n}}{\LINCR{n}} \PAR \OUTPUT{\RINCR{p}}{\LINCR{p}}
\end{array}
\end{math} 
\end{center}
\end{ptrans}

A central element in this translation is the encoding of replication, $\PTRANS[\REPL{P}](n,p)$, so we shall give some further details about its underlying intuitions.
Firstly, with higher-order process mobility, we can create a diverging process simply as $\LIFT{x}{\DCOPY{x}} \PAR \DCOPY{x}$.
This construction is reminiscent of the \LAMBDA-calculus $\Omega$-combinator $(\lambda x.xx) \lambda x .xx$:
$\DCOPY{x}$ will run the process it receives on $x$ whilst simultaneously making it available again on $x$, so by sending it a copy of $\DCOPY{x}$ itself, we obtain a process that continuously copies itself.
Then, by embedding another process $P$ in this construct, $\LIFT{x}{P \PAR \DCOPY{x}} \PAR \DCOPY{x}$, we obtain a process that will create arbitrarily many copies of $P$ at runtime.
Thus we can implement unguarded replication by using just a single name $x$.
However, this name $x$ must not be used by any other process, lest it might interfere with the replication.
This is achieved in the above encoding by composing the two name parameters, $n$ and $p$, to obtain a new name $\NCOMP{n}{p}$.

Secondly, if $\PTRANS[P](n,p)$ were simply copied in this fashion, any usage of the parameters $n$ and $p$ within the translation of $P$ would also be copied, which thus could create a name clash.
Therefore, the inner process is prefixed with two inputs that \emph{bind} $n$ and $p$ within the continuation.
In parallel, we then have two other processes, $\OUTPUT{\RINCR{n}}{\LINCR{n}}$ and $\OUTPUT{\RINCR{p}}{\LINCR{p}}$, that output the new names $\LINCR{n}$ and $\LINCR{p}$, which will be substituted for $n$ and $p$.
These processes are also copied, and in the next round of replication they will instead create the names $\QUOTE{\OUTPUT{\LINCR{n}}{\LINCR{n}}}$ and $\QUOTE{\OUTPUT{\LINCR{p}}{\LINCR{p}}}$, and so on, thereby implementing a \emph{runtime} form of name generation, similar to our static quoting technique.\note{The two previous paragraphs are new, at the request of R1.}

For the purpose of defining a notion of behavioural equivalence that is comparable to that of other calculi that do feature a $\nu$-operator, Meredith and Radestock define a \emph{name-restricted observation predicate} $\BARBSYM^{\SETNAME{N}}$ for the \RHO-calculus, parametrised with a set of names $\SETNAME{N}$.
The idea is to only allow observation of names in this set.
We follow their definition, but also allow the observation predicate to distinguish between input $x$, and output $\OUT{x}$:\footnote{The added distinction between input and output observations is only for use in our later development of a correct encoding, and does not invalidate our claim that the encoding by Meredith and Radestock is incorrect, since our counter-examples shall only rely on observing outputs.}
\begin{center}
\begin{semantics}
  \RULE[\RHO-bOut]
    { x_1 \NAMEEQ x_2 \AND x_1 \in \SETNAME{N} }
    { \LIFT{x_1}{P} \BARBSYM^{\SETNAME{N}} \OUT{x_2} }
\end{semantics}
\begin{semantics}
  \RULE[\RHO-bIn][rho:bin]
    { x_1 \NAMEEQ x_2 \AND x_1 \in \SETNAME{N} }
    { \INPUT{x_1}{y}.P \BARBSYM^{\SETNAME{N}} x_2 }
\end{semantics}
\begin{semantics}
  \RULE[\RHO-bPar][rho:bpar]
    { P_1 \BARBSYM^{\SETNAME{N}} \widehat{x} \quad\lor\quad P_2 \BARBSYM^{\SETNAME{N}} \widehat{x} }
    { P_1 \PAR P_2 \BARBSYM^{\SETNAME{N}} \widehat{x} }
\end{semantics}
\end{center}

\noindent where $\widehat{x}$ ranges over $x, \OUT{x}$.
An \emph{$\SETNAME{N}$-restricted barbed bisimulation} is then a symmetric, binary relation $\SETNAME{R}^{\SETNAME{N}}$ on processes, parametrised with a set of names $\SETNAME{N}$, such that $(P_1, P_2) \in \SETNAME{R}^{\SETNAME{N}}$ implies:
\begin{itemize}
  \item If $P_1 \trans P_1'$ then there exists a $P_2'$ such that $P_2 \trans P_2'$ and $(P_1', P_2') \in \SETNAME{R}^{\SETNAME{N}}$.
  \item If $P_1 \BARBSYM^{\SETNAME{N}} \widehat{x}$ then $P_2 \BARBSYM^{\SETNAME{N}} \widehat{x}$.
\end{itemize}

\noindent We say that $P_1$ is \emph{$\SETNAME{N}$-restricted barbed bisimilar} to $P_2$, written $\BISIM^{\SETNAME{N}}$, if there exists an $\SETNAME{N}$-restricted barbed bisimulation $\SETNAME{R}^{\SETNAME{N}}$ such that  $(P_1, P_2) \in \SETNAME{R}^{\SETNAME{N}}$.
The corresponding `weak' observation predicate is then written 
\begin{equation*}
  P \WBARBSYM^{\SETNAME{N}} \widehat{x} \DEFSYM \exists P' \SUCHTHAT P \trans* P' \land P' \BARBSYM^{\SETNAME{N}} \widehat{x}
\end{equation*}

 \noindent where $\trans*$ is the reflexive and transitive closure of $\trans$, and by replacing $P_2 \BARBSYM^{\SETNAME{N}} \widehat{x}$ with $P_2 \WBARBSYM^{\SETNAME{N}} \widehat{x}$, and $P_2 \trans P_2'$ with $P_2 \trans* P_2'$ in the above definition, we obtain the corresponding notion of a \emph{weak} $\SETNAME{N}$-restricted barbed bisimulation.
We say that $P_1$ is \emph{weakly $\SETNAME{N}$-restricted barbed bisimilar} to $P_2$, written $\WBISIM^{\SETNAME{N}}$, if there exists a weak $\SETNAME{N}$-restricted barbed bisimulation $\SETNAME{R}^{\SETNAME{N}}$ that relates them.

The corresponding observation predicate for the \PI-calculus is built by the following rules for observation on output, restriction and replication 
\begin{center}
\begin{semantics}
  \RULE[\PI-bOut]
    { x \in \SETNAME{N} }
    { \OUTPUT[pi]{x}{y} \BARBSYM^{\SETNAME{N}} \OUT{x} }
\end{semantics}
\begin{semantics}
  \RULE[\PI-bRes](x \neq z)
    { P \BARBSYM^{\SETNAME{N}} \widehat{x} }
    { \NEW{z}P \BARBSYM^{\SETNAME{N}} \widehat{x} }
\end{semantics}
\begin{semantics}
  \RULE[\PI-bRep]
    { P \BARBSYM^{\SETNAME{N}} \widehat{x} }
    { \REPL{P} \BARBSYM^{\SETNAME{N}} \widehat{x} }
\end{semantics}
\end{center}

\noindent and with rules similar to \nameref{rho:bpar} and \nameref{rho:bin} in the \RHO-calculus for observation on parallel composition and input, with strict syntactic equality replacing name equivalence in the premise of the latter rule.
The notions of a weak observation predicate, and (strong resp.\@ weak) $\SETNAME{N}$-restricted barbed bisimulation and bisimilarity for the \PI-calculus are then defined as in the \RHO-calculus.
We write $P \BARBSYM \widehat{x}$, $P \WBARBSYM \widehat{x}$, $P_1 \BISIM P_2$ and $P_1 \WBISIM P_2$ when $\SETNAME{N}$ is the set of all names, corresponding to no restriction on the names we can observe.
This yields the familiar notions of (strong resp.\@ weak) barbed bisimilarity in the \PI-calculus (as defined in e.g.\@ \cite{PICALC}).\note{I added the above to clarify that barbed bisimilarity in the \PI-calculus is not the same as in the \RHO-calculus.}

Given these notions of behavioural equivalence, Meredith and Radestock then state the following as a theorem \cite[p.\@ 14, Theorem 5.3]{RHOCALC}, but without providing a proof:\note{Change `equation' to `the claim stated in (1).'}
\begin{equation}\label{eq:mrtrans}
  P_1 \WBISIM P_2 \iff \PTRANS[P_1] \WBISIM^{\FN{P_1} \UNION \FN{P_2}} \PTRANS[P_2]
\end{equation}

\noindent with observation in the \RHO-calculus restricted to $\FN{P_1} \UNION \FN{P_2}$, i.e.\@ the free names in $P_1$ and $P_2$, implemented as \RHO-names.\footnote{%
Note that the original presentation \cite[p.\@ 14, Theorem 5.3]{RHOCALC} only has $P_1 \WBISIM P_2 \iff \PTRANS[P_1] \WBISIM^{\FN{P_1}} \PTRANS[P_2]$, but we regard this as a simple omission, since it trivially would not hold for the implication from right to left: 
Take for example $P_1 \DEFSYM \OUTPUT{x}{z}$ and $P_2 \DEFSYM \OUTPUT{x}{z} \PAR \OUTPUT{w}{z}$.
Then we have that $\FN{P_1} = \SET{x}$, and indeed $\PTRANS[P_1] \WBISIM^{\SET{x}} \PTRANS[P_2]$ since for $i \in \SET{1, 2}$ we have that $\PTRANS[P_i] \not\trans$ and $\PTRANS[P_i] \BARBSYM^{\SET{x}} \OUT{x}$; but obviously $P_1 \not\WBISIM P_2$, since $P_2 \BARBSYM \OUT{w}$ but $P_1 \not\BARBSYM \OUT{w}$.
}\note{I added this footnote and $\protect\UNION \protect\FN{P_2}$ to the equation.}

\section{The errors}\label{sec:errors}
We shall now see why the claim stated in~\ref{eq:mrtrans} does not hold.
Firstly, consider the following \PI-calculus processes:
\begin{formatcalculus}[pi]
\begin{equation*}
  P_1 \DEFSYM \REPL{\NEW{z}\OUTPUT{u}{z}} \qquad\text{and}\qquad P_2 \DEFSYM \NEW{z}\REPL{\OUTPUT{u}{z}}
\end{equation*}

Clearly, they represent different behaviours: $P_2$ will continuously send out the \emph{same} fresh name $z$ on $u$, whilst $P_1$ will send out \emph{different} fresh names, as we can see by applying \ALPHA-conversion after unfolding the replication (see \cite[p.\@ 11]{rhocalc2022techreport} for details).
We can also easily construct a testing context $C$ such that they can be distinguished by the (\PI-calculus) $\WBARBSYM \OUT{x}$ predicate, for example
\begin{equation*}
  C \DEFSYM \HOLE[~] \PAR \INPUT{u}{n_1} . \INPUT{u}{n_2} . \PAREN{ \OUT{n_1} \PAR n_2 . \OUT{x} }
\end{equation*}

\noindent where the objects for the input/output of $\OUT{n_1}, n_2$ and $\OUT{x}$ are ignored, as this only requires pure synchronisation.
Clearly, if the two names received on $u$ are the same, then $n_1$ and $n_2$ will be the same name, so they can synchronise and we will therefore be able to observe $\OUT{x}$ after 3 reduction steps.
And conversely, if the two names are distinct, then we will not observe $\OUT{x}$.
Thus $C\HOLE[P_1] \not\WBARBSYM \OUT{x}$ whilst $C\HOLE[P_2] \WBARBSYM \OUT{x}$ as argued above.

Now we make a slight adjustment to the two terms.
By composing an arbitrary process $Q$ with the inner output process $\OUTPUT{u}{z}$ we obtain the following:
\begin{equation*}
  P_1' \DEFSYM \REPL{\PAREN{\NEW{z} \OUTPUT{u}{z} \PAR Q }} \quad\text{and}\quad
  P_2' \DEFSYM \NEW{z}\REPL{\PAREN{ \OUTPUT{u}{z} \PAR Q }}
\end{equation*}

The actual behaviour of $Q$ is irrelevant; it is there solely to induce the parameter pair $(n, p)$ to be split into a `left pair' $(\LINCR{n}, \LINCR{p})$ and a `right pair' $(\RINCR{n}, \RINCR{p})$ that are passed to the translations of the left (resp.\@ right) parts of the parallel composition.
Note also that this changes nothing w.r.t.\@ observability of $\OUT{x}$: we still have that $C\HOLE[P_1'] \not\WBARBSYM \OUT{x}$ and $C\HOLE[P_2'] \WBARBSYM \OUT{x}$.

We shall now perform the actual translation. 
To make it more readable, we tabulate the names generated by static quoting during the translation and rename them as follows:
\begin{center}
\begin{math}
\begin{array}{r @{~} l @{~~~}   r @{~} l @{~~~}   r @{~} l @{~~~}  r @{~} l @{~~~}   r @{~} l  }
  \NCOMP{n}{p} & = a   & \RINCR{p}    & = c   & \LINCR{p}            & = e   & \RINCR{(\LINCR{p})}              & = g   & \LINCR{\LINCR{n}}   & = i \\
  \RINCR{n}    & = b   & \LINCR{n}    & = d   & \RINCR{(\LINCR{n})}  & = f   & \NCOMP{(\LINCR{n})}{(\LINCR{p})} & = h   & \LINCR{\LINCR{p}}   & = j 
\end{array} 
\end{math} 
\end{center}
\end{formatcalculus}

\noindent Note that \emph{none} of these names will be observable by the $\WBARBSYM^{\FN{P_1} \UNION \FN{P_2}}$ predicate, because they are generated by the translation, and hence are not in the set $\FN{P_1} \UNION \FN{P_2}$ of free names of $P_1$ and $P_2$.
Now, here is the translation:
\begin{formatcalculus}[rho]
\begin{align*}
  \PTRANS[P_1'](n,p)  & = \LIFT{a}{ \INPUT{b}{n} . \INPUT{c}{p} . \PAREN[Big]{\INPUT{e}{z} . \OUTPUT{u}{z} \PAR \OUTPUT{e}{d} \PAR \PTRANS[Q](f,g) \PAR \DCOPY{a} \PAR \LIFT{b}{\OUTPUT{n}{n}} \PAR \LIFT{c}{\OUTPUT{p}{p}} } } \\ 
                      & ~\PAR \DCOPY{a} \PAR \OUTPUT{b}{d} \PAR \OUTPUT{c}{e} \\
  \PTRANS[P_2'](n,p)  & = \INPUT{p}{z} . \LIFT{h}{ \INPUT{f}{d} . \INPUT{g}{e} . \PAREN[Big]{ \OUTPUT{u}{z} \PAR \PTRANS[Q](f,g) \PAR \DCOPY{h} \PAR \LIFT{f}{\OUTPUT{d}{d}} \PAR \LIFT{g}{\OUTPUT{e}{e}} } } \\
                      & ~\PAR \DCOPY{h} \PAR \OUTPUT{f}{i} \PAR \OUTPUT{g}{j} \PAR \OUTPUT{p}{n}
\end{align*}
\end{formatcalculus}

By performing the reductions, we see (not surprisingly) that $\PTRANS[P_2'](n,p)$ firstly performs the communication on $p$, which causes $z$ to be replaced by $n$, and the process afterwards expands into arbitrarily many instances of $\OUTPUT{u}{n}$ (see~\cite[p.\@ 12]{rhocalc2022techreport} for a reduction sequence).
On the other hand, the translated process $\PTRANS[P_1'](n,p)$ will immediately go through the replication steps, thereby creating arbitrarily many instances of the process $\INPUT{e}{z} . \OUTPUT{u}{z} \PAR \OUTPUT{e}{d}$ corresponding to the translation of $\NEW{z}\OUTPUT[pi]{u}{z}$.
This process obviously reduces to $\OUTPUT{u}{d}$ in one step.
However, precisely because of the aforementioned split of $(n,p)$ over the translation of parallel composition, the name $d$ will \emph{not} be updated by the replication context. 
This process will therefore \emph{also} repeatedly output the \emph{same} name $d$ on $u$, and the (translated) form of our testing context can therefore no longer distinguish the processes.

Both $\PTRANS[P_1']$ and $\PTRANS[P_2']$ thus reduce to arbitrarily many copies of either $\OUTPUT{u}{d}$ (for $P_1'$) or $\OUTPUT{u}{n}$ (for $P_2'$), and $u$ is the only name we can observe, as all the other names are created by the translation.
This then gives us our desired counter-example: by also translating the testing context we obtain a pair of processes where
\begin{equation*}
  C\HOLE[P_1'] \not\WBISIM C\HOLE[P_2'] \qquad\text{but}\qquad \PTRANS[{C\HOLE[P_1']}] \WBISIM^{\FN{C\HOLE[P_1']} \UNION \FN{C\HOLE[P_2']} } \PTRANS[{C\HOLE[P_2']}]
\end{equation*}

\noindent in contradiction of the implication from right to left in the claim stated in~\ref{eq:mrtrans}.

The detailed analysis above gives us a clear idea of the root cause of the problem: 
The translation of replication creates a context with the purpose of ensuring that the names $(n,p)$ used within it will repeatedly be substituted with new, fresh names $(\LINCR{n},\LINCR{p})$ dynamically built from the previous names, and these act as sources of new names for any occurrence of $\NEW{z}P$ within a replicated process. 
The point is precisely to ensure that each instance of a replicated $\nu$ operator will generate a unique new name, and the parameters $(n,p)$ on the translation function act as `handles' to access this resource; they are the names that have \emph{most recently} been replicated.

The problem arises because this property of being the `most recently replicated names' is not preserved by the translation of parallel composition: 
It splits the pair into a left and a right pair, used in the translation of the left and right parallel components:
\begin{equation*}
    \PTRANS[P_1 \PAR P_2](n,p) = \PTRANS[P_1](\LINCR{n}, \LINCR{p}) \PAR \PTRANS[P_2](\RINCR{n}, \RINCR{p})
\end{equation*}

\noindent Thus, the access to the most recently replicated names is lost in the translation of the inner processes, because, as we noted above, substitution does \emph{not} recur into processes under quotes.
Therefore, when the replication context increments $(n,p)$ at runtime, this update cannot touch the $n$ and $p$ embedded in the \emph{statically} incremented names $(\LINCR{n},\LINCR{p})$ and $(\RINCR{n},\RINCR{p})$ which the translation function generates for the translation of parallel composition.
This is why we added an arbitrary $Q$ to create a parallel composition in our counter-example above.

However, the error above is not the only one in the claim by Meredith and Radestock: whilst its root cause was the splitting of names over the translation of parallel composition, we can also create another example that is more directly related to the interplay between $\NEW{x}P$ and replication.
Consider the following processes:
\begin{equation*}
  P_1 \DEFSYM \REPL{\NEW{z}\OUTPUT[pi]{u}{z}}         \qquad\text{and}\qquad
  P_2 \DEFSYM \REPL{\NEW{q}\NEW{z}\OUTPUT[pi]{u}{z}}
\end{equation*}

Note that $P_1$ and $P_2$ are structurally congruent, since the new name $q$ is never used.
Thus $P_1 \WBISIM P_2$ also holds.
Yet when we translate those terms, the name incrementation in the translation of a term of the form $\NEW{x}P$ means that we again lose access to the most recently replicated names from the translation of replication.
This can be easily seen if we perform the translation stepwise, using the same tabulated list of names as before.
For both processes, the translation of replication is the same:
\begin{formatcalculus}[rho]
\begin{equation*}
  \PTRANS[\REPL{P}](n,p) = \LIFT{a}{ \INPUT{b}{n}.\INPUT{c}{p}.\PAREN{ \PTRANS[P](n,p) \PAR \DCOPY{a} \PAR \LIFT{b}{\OUTPUT{n}{n}} \PAR \LIFT{c}{\OUTPUT{p}{p}} } } \PAR \DCOPY{a} \PAR \OUTPUT{b}{d} \PAR \OUTPUT{c}{e}
\end{equation*}
\end{formatcalculus}

Now let $P_1' \DEFSYM \NEW{z}\OUTPUT[pi]{u}{z}$ and $P_2' \DEFSYM \NEW{q}\NEW{z}\OUTPUT[pi]{u}{z}$ and replace $\PTRANS[P](n,p)$ above with $\PTRANS[P_1'](n,p)$ and $\PTRANS[P_2'](n,p)$ respectively.
The translations of the inner processes yield:
\begin{align*}
  \PTRANS[{\NEW{z}\OUTPUT[pi]{u}{z}}](n,p)        & = \INPUT{p}{z}.\OUTPUT{u}{z} \PAR \OUTPUT{p}{n} \\
  \PTRANS[{\NEW{q}\NEW{z}\OUTPUT[pi]{u}{z}}](n,p) & = \INPUT{p}{q}.\PAREN{ \INPUT{e}{z}.\OUTPUT{u}{z} \PAR \OUTPUT{e}{d} } \PAR \OUTPUT{p}{n}
\end{align*}

\noindent which reduce to $\OUTPUT{u}{n}$ and $\OUTPUT{u}{d}$ respectively.
The names $n,p$ are bound in the replication context and will therefore be updated whenever the process replicates.
However, in the case of $P_2$, these names are \emph{statically} incremented in the translation of $\NEW{q}$ to yield the names $\LINCR{n} = d$ and $\LINCR{p} = e$, and these two names will therefore \emph{not} be updated at runtime, just as in the previous counter-example.
Consequently, in the case of $P_2$ the names sent out on $u$ will \emph{not} be distinct; they will all be the name $\LINCR{n} = d$.
We can therefore use the same testing context $C$ as in the previous example and proceed as before to generate another contradiction of the claim in~\ref{eq:mrtrans}; this time by \emph{distinguishing} the translated terms, although we have $C\HOLE[P_1] \WBISIM C\HOLE[P_2]$ in the \PI-calculus.
In summary, neither of the implications in the claim stated in~\ref{eq:mrtrans} hold.

%
%

\section{Our criteria for encodability}
Both of the previous examples illustrate the difficulties involved in reasoning about a parametrised translation.
Usually, the parameters represent a property or invariant that is assumed to be preserved throughout the translation, and a proof of correctness of the translation must therefore also include a proof that this invariant or property is indeed preserved.
For example, in the present case, the invariant assumed to hold for the parameters is precisely that they always refer to the most recently replicated names.
However, this assumption is never formally stated in the original \RHO-calculus paper \cite{RHOCALC}, and as the examples above show, it does not hold either.
Thus, a naive attempt to show correctness of the translation by induction in the clauses of the translation function may therefore seemingly go through, if the parameters are not considered.
This is doubly problematic in the present case, because the observation predicate used in the bisimulation relation over \RHO-calculus terms is parametrised so that we do \emph{not} observe the names created by the translation function.

Full abstraction, of which the claim in~\ref{eq:mrtrans} is an instance, may also not be the most informative correctness criterion, as argued by Gorla and Nestmann \cite{gorla_nestmann2014_full_abstraction}; for example, it does not necessarily prevent the translation from introducing divergence.
Also, as we are here more interested in showing that the \PI-calculus is `implementable' in the \RHO-calculus than in transferring equations between the source and target language, we shall instead follow the approach of such authors as Gorla \cite{GORLA}, Carbone and Maffeis \cite{CARBONEMAFFEIS} and others, and state a number of criteria for what we consider a valid encoding, where we also take the presence of parameters into account:

\begin{definition}[Language]
A \emph{language} $\LANG$ is a tuple $\LANG \DEFSYM (\PROC, \NAMES, \trans, \BEHEQ)$, where $\PROC$ is a set of terms, $\NAMES$ is a set of names, ${\trans} \subseteq \PROC \times \PROC$ is the reduction relation, with $\trans*$denoting the reflexive and transitive closure of $\trans$, and ${\BEHEQ} \subseteq \PROC \times \PROC$ is a notion of behavioural equivalence.
\end{definition}

We say a term $P \in \PROC$ \emph{diverges}, written $P \trans^{\omega}$, if $P$ has an infinite reduction sequence.
We use $\sigma : \NAMES \to \NAMES$ to denote a substitution function in $\LANG$.
For encodings, we need the notion of a \emph{source} and a \emph{target} language, and we shall generally use the convention of subscripting $s$ (for source) and $t$ (for target) to a language $\LANG$ or its components, including substitutions, and we let $S \in \PROC_s$ and $T \in \PROC_t$.

\begin{definition}[Encoding]
An \emph{encoding} of $\LANG_s$ into $\LANG_t$ is a tuple $(\PTRANS[~](N), \PHI, \delta)$, where $\PTRANS[~](N) : \PROC_s \to \PROC_t$ is a \emph{translation function}, parametrised with a finite list of names $N \in \NAMES_t^k$; and $\PHI : \NAMES_s \to \NAMES_t$ is a \emph{renaming policy}, mapping names in the source language into names in the target language; and $\delta : \NAMES_t^k \to \NAMES_t^k$ is a \emph{name derivation} function, mapping $k$-ary tuples of target names to tuples of equal arity for some $k$.
\end{definition}

The name derivation function $\delta$ allows us to express that the list of name parameters $N$ may evolve in some predictable way during the course of translation.
This seems necessary in particular when we are working with a language with structured terms as names.
In some cases we may also need to derive multiple tuples of names from the same input tuple; thus to comply with the requirement that $\delta$ is a single function, we could e.g.\@ envision using an extra, designated name as argument to control the derivation method used by $\delta$.
However, to abstract away from such details, we say that a tuple of names $N_2$ is \emph{derivable} from some tuple of names $N_1$, written $N_1 \leadsto N_2$, if $\delta\ARGS{N_1} = N_2$, and likewise that $N_1 \leadsto n$ if $n \in N_2$.
Note that we abuse the notation slightly and treat the lists as sets when the position of each individual component does not matter.\note{R1: Changed the definition of derivable to make it clear that it is a fixed function.}

\begin{definition}[Valid encoding]\label{def:valid_encoding}
We shall regard an encoding as \emph{valid}, if it satisfies at least the following criteria:
\begin{enumerate}
  \item\label{cond:parallel} Compositionality: $\PTRANS[S_1 \PAR \ldots \PAR S_n](N) = C \PAR \PTRANS[S_1]({N_1}) \PAR \ldots \PAR \PTRANS[S_n]({N_n})$ where $C$ is an optional coordinating context and $\FN{C} \subseteq \PHI\ARGS{\FN{S_1 \PAR \ldots \PAR S_n}} \UNION N$, and for each $i \in \SET{1, \ldots, n}$ we have that $N \leadsto N_i$. 

  \item\label{cond:substitution} Substitution invariance: $\PTRANS[S\ssubst](N) \BEHEQ \PTRANS[S](N)\tsubst$ for each $\ssubst$, where $\PHI\ARGS{\ssubst\ARGS{x}} = \tsubst\ARGS{\PHI\ARGS{x}}$.

  \item\label{cond:operational_correspondence} Operational correspondence: $S \trans* S' \iff \exists T' \SUCHTHAT \PTRANS[S](N) \trans* T' \land T' \BEHEQ \PTRANS[S'](N')$ and $N \leadsto N'$.
   
  \item\label{cond:observational_correspondence} Observational correspondence: We require that $N \INTERSECT \PHI\ARGS{\SETNAME{M}} = \EMPTYSET$ for any set of observable names $\SETNAME{M}$. 
  Then $P \BARBSYM^{\SETNAME{M}} \widehat{x} \iff \PTRANS[P](N) \WBARBSYM^{\PHI\kern2pt\ARGS{\SETNAME{M}}} \PHI\ARGS{\widehat{x}}$.\note{I changed this to emphasise that $N \INTERSECT \PHI\ARGS{\SETNAME{M}}$ is a condition, not a consequence.}

  \item\label{cond:divergence_reflection} Divergence reflection: $\PTRANS[P](N) \trans^{\omega} \implies P \trans^{\omega}$.
  
  \item\label{cond:parameter_independence} Parameter independence: $\PTRANS[P]({N_1}) \BEHEQ \PTRANS[P]({N_2})$ for each finite $N_1, N_2$.
\end{enumerate}
\end{definition}

These criteria are very close to those proposed by Gorla \cite{GORLA}, except that we have chosen observational correspondence, rather than the less specific success testing; i.e.\@ $P \trans* \BARBSYM \SUCCESS$ implies $\PTRANS[P](N) \trans* \BARBSYM \SUCCESS$.
This can easily be obtained, simply by choosing a specific name $x$ and then defining $\SUCCESS$ as a process with $x$ in subject position, as we did in our counter-examples above.

Furthermore, as we are here allowing parameters to appear on the translation, we have also added the criterion of parameter independence, which does not appear in \cite{GORLA}.
This is just to ensure that the behaviour of the translated terms will not depend on the exact choice of the parameters.
Likewise, we have also added name restriction to the observation predicate for observational correspondence $\WBARBSYM^{\SETNAME{M}}$, and we require that $N \INTERSECT \SETNAME{M} = \EMPTYSET$; i.e.\@ that the parameters should not be observable.
This seems a natural requirement, since we also require that $N \subseteq \SETNAME{N}_t$; i.e.\@ that the parameters belong to the target language.
They should therefore not be observable on the source terms.

\section{A correct encoding}\label{sec:pi_encoding}
As the previous examples have illustrated, the main difficulty in creating an encoding of the \PI-calculus in the \RHO-calculus, is how to achieve a robust source of fresh names at runtime that are guaranteed never to cause a name clash.
One way is to use a dedicated process for this purpose.
Consider the following process, where $\DCOPY{x}$ is defined as in section~\ref{sec:mrencoding}:
\begin{equation*}
  \NAMESERVER \DEFSYM \DCOPY{x} \PAR \LIFT{x}{ \INPUT{z}{a}.\INPUT{v}{r}.\PAREN[Big]{ \DCOPY{x} \PAR \LIFT{r}{\DROP{a}} \PAR \LIFT{z}{\LIFT{a}{\NIL}} } } \PAR \LIFT{z}{\DROP{s}}
\end{equation*}

This process is a \emph{name server}; it consistently generates names corresponding to consecutive left-increments of the initial name $s$ and outputs them on the `return address' $r$ received on $v$.
We refer to the above form as the \emph{initial state} of the name server and note that after two reductions it evolves to the form
\begin{equation*}
  \INPUT{v}{r}.\PAREN[Big]{ \DCOPY{x} \PAR \LIFT{r}{\DROP{s}} \PAR \LIFT{z}{\LIFT{\QUOTE{\DROP{s}}}{\NIL}} } \PAR \LIFT{x}{  \INPUT{z}{a}.\INPUT{v}{r}.\PAREN[Big]{ \DCOPY{x} \PAR \LIFT{r}{\DROP{a}} \PAR \LIFT{z}{\LIFT{a}{\NIL}} }  }
\end{equation*}

\noindent which we refer to as its \emph{ready state}, where it blocks, awaiting a request for a new name on $v$.
The first request will return $\QUOTE{\DROP{s}}$; a second request will return $\QUOTE{\LIFT{\QUOTE{\DROP{s}}}{\NIL}} = \LINCR{s}$, and so on.

We can verify that the names will all be distinct by considering the \emph{quote depth} of a name (resp.\@ process) defined thus:
\begin{center}
\begin{math}
\begin{array}{r @{~} l}
  \QDEPTH{\QUOTE{P}} & = %
  \begin{cases}
    \QDEPTH{x}      & \text{if } P \STRUCTEQ \DROP{x} \\
    1 + \QDEPTH{P}  & \text{otherwise}
  \end{cases} 
\end{array}\hspace*{0.3cm}
\begin{array}{r @{~} l}
  \QDEPTH{P} & = %
  \begin{cases}
    \max\SET{ \QDEPTH{x} | x \in \FN{P} } & \text{if } \FN{P} \neq \EMPTYSET \\
    0                                     & \text{otherwise}
  \end{cases}
\end{array}
\end{math}  
\end{center}

\noindent The quote depth of a name $x_1$ corresponds to the maximum number of calls to \nameref{n:struct} used to conclude name equivalence $x_1 \NAMEEQ x_2$ for some name $x_2$.
Thus, a necessary (but not sufficient) condition for two names to be name equivalent is that they have the same quote depth.
Names are therefore automatically stratified based on their quote depth:
\begin{lemma}[Stratification]\label{lemma:stratification}
$x_1 \NAMEEQ x_2 \implies \QDEPTH{x_1} = \QDEPTH{x_2}$.
\end{lemma}

We can also partition names into \emph{namespaces} in the following way: let $\NAMESPACE{~}$ be a collection of name contexts, ranged over by $N$, with one or more holes occurring in the position of free names.
If $s$ is a name, then so is $N\HOLE[s]$ for some $N \in \NAMESPACE{~}$.
We write $\NAMESPACE{s} \DEFSYM \SET{ N\HOLE[s] | N \in \NAMESPACE{~} }$, and we say that $\NAMESPACE{s}$ is a \emph{namespace rooted at $s$}. 
Clearly, if $\QDEPTH{N} = n$ (counting $\QDEPTH{\HOLE} = 0$), and $\QDEPTH{s} = i$ and $\QDEPTH{s'} = j$, then $\QDEPTH{N\HOLE[s]} = n+i$ and $\QDEPTH{N\HOLE[s']} = n+j$.

Using the concepts of name contexts, we can describe our aforementioned three static quoting techniques as three distinct name space `templates,' built by the following grammars:
\begin{align*}
  \LNC \in \LNS{~} & \DCLSYM \HOLE[~] \ORSYM \QUOTE{\LIFT{\LNC}{\NIL}} \\
  \RNC \in \RNS{~} & \DCLSYM \HOLE[~] \ORSYM \QUOTE{\INPUT{\RNC}{\QUOTE{\NIL}}.\NIL} \\
  \CNC \in \CNS{~} & \DCLSYM \HOLE[~] \ORSYM \QUOTE{\LIFT{\CNC}{\NIL} \PAR \INPUT{\CNC}{\QUOTE{\NIL}}.\NIL}
\end{align*}

\noindent We shall use these namespace templates to implement the name derivation function $\delta$.
Thus, if we let $\widehat{N}$ denote any of the name contexts $\LNC, \RNC, \CNC$ then $s \leadsto s'$ if there exists a name context $\widehat{N}$ such that $s' \NAMEEQ \widehat{N}\HOLE[s]$.
This assures us that even if two namespaces use the same structure, e.g.\@ $\LNS{~}$, all their names will still be distinct \emph{if} their roots are not name equivalent, and neither is derivable from the other.\note{Changed this to account for new definition of derivability.}

In case of the name server, we see that it generates the namespace $\LNS{s}$, i.e.\@ the namespace of left-increments rooted at $s$, where $s$ is a parameter.
Thus if $s_1 \not\NAMEEQ s_2$ and neither is derivable from the other, then $\NAMESERVER[s_1]$ and $\NAMESERVER[s_2]$ will generate similarly structured namespaces, $\LNS{s_1}$ and $\LNS{s_2}$, but consisting of different sets of names.
Yet we can easily construct a mapping $\LNS{s_1} \mapsto \LNS{s_2}$ simply by replacing $s_1$ with $s_2$ within each name $\LNC\HOLE[s_1] \in \LNS{s_1}$.
This will be important in the proof for parameter independence below.\note{R1: What is the purpose of this mapping? I added this line as answer.}

Based on these considerations we can now construct our encoding.
We let the encoding be defined as $\PTRANS[P] \DEFSYM \PTRANS[P](n,v) \PAR \NAMESERVER[s]$, where we assume we can choose the names $n,v,x,z,s$ such that they are distinct from all free names in $P$ and $n,v,x,z \not\in \LNS{s}$.
As in the encoding by Meredith and Radestock, we shall assume that all \PI-calculus names are implemented as \RHO-names, and thus we shall generally omit explicit reference to $\PHI$ in the following.
We shall also limit ourselves to the \PI-calculus fragment with only input-guarded replication, to ensure that the encoding does not introduce divergence, unlike the encoding by Meredith and Radestock which replicates eagerly and therefore always diverges.\footnote{This is only a slight limitation, as we can use input-guarded replication to encode full replication.
Note also that having only input-guarded replication would not have prevented any of the errors described in section~\ref{sec:errors}.}
This can be achieved by prefixing the object of the lift with an input construct, i.e.\@ $\LIFT{n}{\INPUT{x}{y}.\PAREN{\DCOPY{n} \PAR P}}$, since
\begin{equation*}
  \DCOPY{n} \PAR \LIFT{n}{\INPUT{x}{y}.\PAREN{\DCOPY{n} \PAR P}} \trans \INPUT{x}{y}.\PAREN{\DCOPY{n} \PAR P} \PAR \LIFT{n}{\INPUT{x}{y}.\PAREN{\DCOPY{n} \PAR P}}
\end{equation*}

\noindent and the process then blocks until it receives a communication on $x$.
Given these considerations, the translation function $\PTRANS[~](n,v)$ is then given by the following equations:
%
\begin{center}
\begin{math}
\begin{array}{r @{~} l}
  \PTRANS[\NIL](n,v)                  & = \NIL                                                            \\
  \PTRANS[P_1 \PAR P_2](n,v)          & = \PTRANS[P_1]({\LINCR{n},v}) \PAR \PTRANS[P_2]({\RINCR{n},v})    \\ 
  \PTRANS[\INPUT{x}{y}.P](n,v)        & = \INPUT{x}{y} . \PTRANS[P](n,v)                                  \\
\end{array}\hspace{0.5cm}
\begin{array}{r @{~} l}
  \PTRANS[{\OUTPUT[pi]{x}{z}}](n,v)   & = \OUTPUT{x}{z}                                                   \\
  \PTRANS[\NEW{x}P](n,v)              & = \OUTPUT{v}{n} \PAR \INPUT{n}{x} . \PTRANS[P]({\MCOMP{n},v})     \\
  \PTRANS[\REPL{\INPUT{x}{y}.P}](n,v) & = \DCOPY{n} \PAR \LIFT{n}{ \INPUT{x}{y}.\PAREN{ \DCOPY{n} \PAR \PTRANS[P]({\MCOMP{n},v}) } }
\end{array}
\end{math}
\end{center}

The idea is that we simplify the `bookkeeping' involved in runtime name generation by isolating it to a single, contextual process.
This prevents errors of the first kind in the encoding by Meredith and Radestock, which resulted from processes losing access to the most recently replicated names.
Here, the name $v$ is used by all processes to contact the name server, and since it is never updated this access can never be lost.
Conversely, the name $n$, which is used for the `return address,' as well as for replication, is \emph{always} updated incrementally, during the translation.
It is never bound or reused, unlike in the translation by Meredith and Radestock, where the replication context used $\LINCR{n}, \LINCR{p}$ but also bound $n,p$ and passed them to the inner translation of $P$, which resulted in the second kind of error.
We say that a name is \emph{unique for the translation} if it is never generated more than once by the translation function, and this is the invariant that should hold for the parameter $n$: 

\begin{lemma}[Uniqueness]\label{lemma:uniqueness}
For each clause $\PTRANS[{C\HOLE[P]}](n,v) = \PTRANS[C](n,v)\HOLE[{\PTRANS[P](n',v)}]$, where $n \leadsto n'$, and $\PTRANS[C](n,v)$ contains a set of names $N' = \SET{n_1, \ldots, n_k}$ such that $n \leadsto N'$, it holds that if $n$ is unique for the translation, then so are $n_1, \ldots, n_k$ and $n'$.
\end{lemma}

This can easily be shown by examining the clauses of the translation function, assuming $n$ is unique.
For every usage of $n$ in a clause, we always either increase the quote depth of the parameter we pass to the inner call to the translation, or we shift the parameter into a new namespace by composition.
Furthermore, the behaviour of the translated process does not depend on the structure of the name parameter $n$, as long as $n$ is unique:
\begin{proposition}[Independence of parameters]\label{prop:independence}
If $\FRESH{n,n',s,s'}\NAMESOF{P}$ and all are unique for the translation, then $\PTRANS[P](n,v) \PAR \NAMESERVER[s] \BISIM^{\FN{P}} \PTRANS[P](n',v) \PAR \NAMESERVER[s']$. 
\end{proposition}

This follows from the fact that the translation only generates finitely many names, say, of the structure $\NAMESPACE{~}$, so we can construct a finite substitution $\tsubst : \NAMESPACE{n} \to \NAMESPACE{n'}$ and simply apply it to $\PTRANS[P](n,v)$ to obtain $\PTRANS[P](n',v)$.
Then as we know that $\FRESH{n,n',s,s'}\NAMESOF{P}$, and by construction $x \notin \FN{P}$ for each $x \in \NAMESPACE{s} \UNION \NAMESPACE{s'}$, none of these names can be observed by the $\BARBSYM^{\FN{P}}$ predicate, so they cannot be used to distinguish the two processes.
A similar argument can then be made for the name server and the two namespaces $\LNS{s}$ and $\LNS{s'}$ generated by it at runtime.

Next, we formulate a (mostly) standard result relating substitution in the two calculi:
\begin{proposition}[Substitution]\label{prop:substitution}
Let $\ssubst \DEFSYM \REPLACE{u}{w}$ denote substitution in the \PI-calculus, and let $\tsubst \DEFSYM \REPLACE{u}{w}$ denote substitution in the \RHO-calculus.
Then $\PTRANS[P\ssubst](n,v) = \PTRANS[P](n,v)\tsubst$ if $\FRESH{u,w}P,n,v,\NAMESPACE{n}$.
\end{proposition}

This is proved by induction in the clauses of the translation.
The condition $\FRESH{u,w}P,n,v,\NAMESPACE{n}$ ensures that the substitution cannot touch any of the names created by the translation, which is reasonable, since the substitutions we care about should derive from communications in the \PI-calculus, and not from some of the `internal' reductions in the \RHO-calculus that are used to simulate replication or requests for new names.

Our next result establishes that our translation preserves observability of subjects, as long as we restrict observations to the set of free names in $P$:
\begin{proposition}[Weak observational correspondence]\label{prop:observational_correspondence}
Let $\WBARBSYM^{\SETNAME{N}}$ be the least predicate such that $\PTRANS[S] \WBARBSYM^{\SETNAME{N}} \widehat{n}$ holds if either of the following conditions are satisfied:
\begin{enumerate}
  \item if $S = S_1 \PAR S_2$ and $\PTRANS[S_1] \WBARBSYM^{\SETNAME{N}} \widehat{n} \lor \PTRANS[S_2] \WBARBSYM^{\SETNAME{N}} \widehat{n}$
  \item if $S \neq S_1 \PAR S_2$ and $\PTRANS[S](n,v) \PAR \NAMESERVER \trans* T' \land T' \BARBSYM^{\SETNAME{N}} \widehat{n}$
\end{enumerate}

\noindent Then for any $x$, $P \BARBSYM^{\FN{P}} \widehat{x} \iff \PTRANS[P] \WBARBSYM^{\FN{P}} \widehat{x}$.\note{Put the conditions into an enumerate to make it more readable}
\end{proposition}

This is proved by induction in the clauses of the translation.
Note that we purposefully restrict the weak observation predicate to only allow reductions involved in replication and requesting a fresh name from the name server; i.e.\@ by splitting it directly over parallel compositions rather than allowing them to first interact. 
This is necessary for proving the implication from right to left in Proposition~\ref{prop:observational_correspondence}, since reductions might otherwise expose more names that are not immediately observable in the source terms.
This restriction can be lifted if we replace $\BARBSYM$ by $\WBARBSYM$ in the \PI-calculus, but we prefer this slightly more complicated formulation to illustrate that observability is strictly preserved, in the sense that any auxiliary steps required in the \RHO-calculus are `internal,' deriving either from a replication step, a request for a new name, or from the name server as it moves from its initial state to its ready state, and neither of these are observable by the $\BARBSYM^{\FN{P}}$ predicate.

Next we show that the translation preserves the semantic meaning of the source program:
\begin{proposition}[Operational correspondence]
$P \trans* P' \iff \PTRANS[P] \trans*\RHOWBISIM \PTRANS[P']$.
\end{proposition}

This proof can be split into two parts.
For the forward direction (completeness) we can actually show the stronger statement that $P \trans P' \implies \PTRANS[P] \trans\trans*\RHOBISIM \PTRANS[P']$ by induction in the reduction semantics of the \PI-calculus, as every reduction in the \PI-calculus is matched by one or more steps in the \RHO-calculus.
The proof often relies on Proposition~\ref{prop:substitution} for the cases of communication, replication and $\NEW{x}P$, and on Proposition~\ref{prop:independence} when we translate the reduct of the \PI-calculus term as this often induces a slightly different form on the parameters.

For the other direction (soundness) we can only prove the weaker form $\PTRANS[P] \trans* T' \implies \exists P' . P \trans* P' \land T' \RHOWBISIM \PTRANS[P']$, due to the extra reductions deriving from the name server, replication, or requests for new names.
Thus we proceed by induction in the reduction sequence, and often again making use of Proposition~\ref{prop:independence}.

Having only the weaker form of completeness, with $\trans*$ instead of $\trans$, of course means that this statement in itself is not enough to verify that the translation does not introduce divergence.
We therefore prove this separately:
\begin{proposition}[Divergence reflection]\label{prop:divergence_reflection}
$\PTRANS[P] \trans^{\omega} \implies P \trans^{\omega}$.
\end{proposition}

We show this by induction in the clauses of the translation function.
The matter is made easier by the fact that a reduction sequence related to the name server, requests for new names, or unfolding replication, is always of finite length: the name server takes two steps to evolve from its initial state to its ready state, where it blocks until it receives a request; serving a request requires two steps, and then two further steps to return to its ready state; and input-guarded replication takes a single step to unfold once, after which it blocks until it receives an input.

\section{A separation result}
The \RHO-calculus can encode the \PI-calculus, as we saw in the previous section.
However, the converse does not hold.
Under some general assumptions about the behavioural equivalence $\BEHEQ$ used in the target language, we can show that there cannot be an encoding of the \RHO-calculus into the \PI-calculus that satisfies our validity criteria from Definition~\ref{def:valid_encoding}.
This result relies on a simple observation about substitution in the \PI-calculus, namely that reduction is preserved under substitution:
\begin{lemma}\label{lemma:pi_substitution_reduction}
Let $\tsubst = \REPLACE{x}{n}$ be a substitution in the \PI-calculus, with $n \in \FN{P}$ and $\FRESH{x}P$.
Then $P \trans P' \implies P\tsubst \trans P'\tsubst$.
\end{lemma}

This can easily be shown by induction in the semantic rules, and then with an extra induction in structural congruence for the \RULENAME{\PI-struct} rule.

Next, we consider our requirements for the notion of behavioural equivalence:
First of all, $\BEHEQ$ should obviously be an equivalence relation.
Secondly, it should in some sense preserve the semantics of the processes it equates: as we are here working in a reduction system, it should at least preserve reductions and observability, and it should be preserved under substitution:
\begin{definition}[Behavioural equivalence requirements]\label{def:beheq_requiements}
We require that $\BEHEQ$ be at least an equivalence relation over \PI-terms satisfying the following:\note{It wrongly said \protect\RHO-terms here before. Fixed.}
\begin{enumerate}
  \item\label{req:trans} $P_1 \BEHEQ P_2 \land P_1 \trans* P_1' \implies \exists P_2' \SUCHTHAT P_2 \trans* P_2' \land P_1' \BEHEQ P_2'$
  \item\label{req:obs} $P_1 \BEHEQ P_2 \land P_1 \WBARBSYM \widehat{x} \implies P_2 \WBARBSYM \widehat{x}$
  \item\label{req:subst} $P_1 \BEHEQ P_2 \implies P_1\tsubst \BEHEQ P_2\tsubst$
\end{enumerate}
\end{definition}

The requirements suggest that $\BEHEQ$ should be at least weak, barbed congruence, which does not seem too demanding.
However, we prefer to keep the formulation general, without committing to one specific notion of behavioural equivalence, to emphasise that other, stronger choices are also possible.
The following result will then hold for any such choice:

\begin{theorem}[Separation]\label{thm:separation}
If $\BEHEQ$ satisfies the requirements of Definition~\ref{def:beheq_requiements}, then there is no encoding of the \RHO-calculus into the \PI-calculus satisfying the criteria of Definition~\ref{def:valid_encoding}.
\end{theorem}

\begin{proof}
Assume to the contrary that there exists a translation $\PTRANS[~](N) : \PROC[rho] \to \PROC[pi]$ satisfying the criteria of Definition~\ref{def:valid_encoding}.
We show that this leads to a contradiction.\note{I reordered and expanded the text of this proof for R1}
Firstly, let $u \DEFSYM \QUOTE{\DROP{x_1} \PAR \DROP{x_2}}$, and consider the processes $P$ and $P'$ where 
\begin{equation*}
  P   \DEFSYM P_1 \PAR P_2 \qquad
  P_1 \DEFSYM \LIFT{a}{\DROP{x_1} \PAR \DROP{x_2}} \qquad
  P_2 \DEFSYM \INPUT{a}{n} . \LIFT{n}{\NIL} \qquad
  P'  \DEFSYM \LIFT{u}{\NIL}
\end{equation*}

\noindent Thus $P = \LIFT{a}{\DROP{x_1} \PAR \DROP{x_2}} \PAR \INPUT{a}{n} . \LIFT{n}{\NIL}$ and clearly $P \not\BARBSYM u$ and $u \notin \FN{P}$, but $P \trans \LIFT{\QUOTE{\DROP{x_1} \PAR \DROP{x_2}}}{\NIL} = \LIFT{u}{\NIL} = P'$ and $P' \BARBSYM u$.

Consider now the substitution $\tsubst \DEFSYM \REPLACE{m}{\PHI\ARGS{u}}$ for some fresh name $m$, i.e.\@ $m \neq \PHI\ARGS{u}$ and with $m \notin \FN{\PTRANS[P](N)}$.
$P \trans P'$ gives, by criterion~\ref{cond:operational_correspondence} (operational completeness), that $\PTRANS[P](N) \trans* T'$ and $T' \BEHEQ \PTRANS[P'](N')$ for some $T'$ and $N'$ derivable from $N$.
By criterion~\ref{cond:substitution} (substitution invariance), $\tsubst\ARGS{\PHI\ARGS{u}} = m$ implies $\exists \ssubst . \PHI\ARGS{\ssubst\ARGS{u}} = m$, so we can combine $\ssubst$ with the observability predicate.
By criterion~\ref{cond:observational_correspondence} (observational correspondence), since $P' \not\BARBSYM \ssubst\ARGS{u}$, we therefore also have that $\PTRANS[P'](N) \not\WBARBSYM m$.
This establishes that 
\begin{equation*}
  \PTRANS[P](N) \trans* T' \land T' \BEHEQ \PTRANS[P'](N') \land \PTRANS[P'](N') \not\WBARBSYM m
\end{equation*}

\noindent as expected.
By requirement~\ref{req:obs} in Definition~\ref{def:beheq_requiements}, since $\PTRANS[P'](N') \BEHEQ T'$, it must therefore also be the case that $T' \not\WBARBSYM m$, and hence that $\PTRANS[P](N) \not\WBARBSYM m$.

Now consider the term $\PTRANS[P](N)\tsubst$:
Lemma~\ref{lemma:pi_substitution_reduction} yields $\PTRANS[P](N)\tsubst \trans* T'\tsubst \BEHEQ \PTRANS[P'](N')\tsubst$, and by criterion~\ref{cond:substitution} (substitution invariance) $\PTRANS[P'](N')\tsubst \BEHEQ \PTRANS[P'\ssubst]$.
As we know that $P' \BARBSYM u$, this implies that $P'\ssubst \BARBSYM \ssubst\ARGS{u}$, which again implies that $\PTRANS[P'\ssubst](N') \WBARBSYM \tsubst\ARGS{\PHI\ARGS{u}}$, which implies $\PTRANS[P'](N')\tsubst \WBARBSYM m$.
This establishes that
\begin{equation*}
  \PTRANS[P](N)\tsubst \trans* T'\tsubst \land T'\tsubst \BEHEQ \PTRANS[P'](N')\tsubst \land \PTRANS[P'](N')\tsubst \WBARBSYM m
\end{equation*}

\noindent again, as expected.
By requirement~\ref{req:obs} in Definition~\ref{def:beheq_requiements}, since $\PTRANS[P'](N')\tsubst \BEHEQ T'\tsubst$, it must therefore also be the case that $T'\tsubst \WBARBSYM m$, and hence that $\PTRANS[P](N)\tsubst \WBARBSYM m$.

However, consider now the effect of applying the substitution $\PTRANS[P](N)\tsubst$.
By criterion~\ref{cond:parallel} (compositionality), we have that
\begin{equation*}
  \PTRANS[P](N)\tsubst
  = C\tsubst \PAR \PTRANS[P_1]({N_1})\tsubst \PAR \PTRANS[P_2]({N_2})\tsubst 
  = C \PAR \PTRANS[P_1]({N_1})\tsubst \PAR \PTRANS[P_2]({N_2})\tsubst 
\end{equation*}

\noindent where we can eliminate the substitution from $C$, since $\PHI\ARGS{u} \notin N \UNION N_1 \UNION N_2$, as this immediately would violate criterion~\ref{cond:parameter_independence} (parameter independence); and as we know that $u \notin \FN{P}$, we therefore also know that $\PHI\ARGS{u} \notin \FN{C}$, since $C$ at most can contain a subset of the (\PHI-translated) free names of the process and the parameters.
Thus the substitution has no effect on $C$.

Now consider the two subterms $\PTRANS[P_1]({N_1})\tsubst$ and $\PTRANS[P_2]({N_2})\tsubst$.
By criterion~\ref{cond:substitution}, $\PTRANS[P_1]({N_1})\tsubst \BEHEQ \PTRANS[P_1\ssubst](N_1)$ and $\PTRANS[P_2]({N_2})\tsubst \BEHEQ \PTRANS[P_2\ssubst](N_2)$, but when we apply the substitution, we get that
\begin{equation*}
  P_1\ssubst = \PAREN{\LIFT{a}{\DROP{x_1} \PAR \DROP{x_2}}}\ssubst = \LIFT{a}{\DROP{x_1} \PAR \DROP{x_2}} = P_1 \qquad
  P_2\ssubst = \PAREN{\INPUT{a}{n} . \LIFT{n}{\NIL}}\ssubst = \INPUT{a}{n} . \LIFT{n}{\NIL} = P_2
\end{equation*}

\noindent since obviously $u \notin\FN{P_1}$ and $u \notin\FN{P_2}$, so the substitution has no effect on any of the subterms.
Thus 
\begin{equation*}
  C \PAR \PTRANS[P_1\ssubst]({N_1}) \PAR \PTRANS[P_2\ssubst]({N_2}) = C \PAR \PTRANS[P_1]({N_1}) \PAR \PTRANS[P_2]({N_2})
\end{equation*}

\noindent and hence $\PTRANS[P\ssubst](N) = \PTRANS[P](N)$.
By criterion~\ref{cond:substitution} (substitution invariance) $\PTRANS[P\ssubst](N) \BEHEQ \PTRANS[P](N)\tsubst$, and thus we have that $\PTRANS[P](N)\tsubst \BEHEQ \PTRANS[P](N)$.
This then yields the desired contradiction, since, as established above, $\PTRANS[P](N)\tsubst \WBARBSYM m$ but $\PTRANS[P](N) \not\WBARBSYM m$, whilst by requirement~\ref{req:obs} of Definition~\ref{def:beheq_requiements} it must hold that $\PTRANS[P](N)\tsubst \BEHEQ \PTRANS[P](N) \land \PTRANS[P](N)\tsubst \WBARBSYM m \implies \PTRANS[P](N) \WBARBSYM m$.
\end{proof}

The above proof exploits the reflective capability of the \RHO-calculus to create new, \emph{free} names at runtime, which are therefore also observable and substitutable.
Thus, a substitution can affect the reduct of a process, without affecting the process itself, if the reduction step creates a new name.
This cannot be mimicked in the \PI-calculus, where names have no structure and cannot be composed at runtime.
Any new \emph{free} name appearing at runtime can therefore only come from the translation parameters, since it cannot come from the source term; but this would then violate the criterion of parameter independence, since we would then have to choose the parameters such that they correspond to the names that will be created at runtime.

This result does not directly depend on the higher-order characteristics of the \RHO-calculus, and adding higher-order behaviour to the \PI-calculus would not suffice to enable it to encode the \RHO-calculus.
In \cite{HOPICALC}, Sangiorgi gave an encoding of the Higher-Order \PI-calculus, \HOPI{}, in the \PI-calculus.
His encoding also satisfies our criteria from Definition~\ref{def:valid_encoding}, and we therefore also have the following result:
\begin{corollary}
There is no encoding of the \RHO-calculus into \HOPI{} satisfying the criteria of Definition~\ref{def:valid_encoding}, when $\BEHEQ$ satisfies the requirement in Definition~\ref{def:beheq_requiements}.
\end{corollary}

Indeed, if such an encoding existed, we could compose it with the encoding of \HOPI{} into the \PI-calculus, to obtain an encoding of the \RHO-calculus into the \PI-calculus, in contradiction of Theorem~\ref{thm:separation}.
This also indicates that the key feature of the \RHO-calculus which cannot be represented in the \PI-calculus, is not its higher-order characteristics per se, but rather its capability for reflection, which gives it higher-order characteristics as a by-product.

\section{Related works}
The issues of encodability and assessing the relative expressiveness of various process calculi has been considered by several authors; in particular, Gorla \cite{GORLA} proposed a framework for reasoning about encodability and separation w.r.t.\@ a set of criteria that also served as inspiration for the criteria used in the present paper.
Towards the end of the paper, Gorla also discusses some of the difficulties involved in formulating a \emph{general} framework for encodability in the presence of parameters, which particularly pertain to the question of which language the names belong to (the source or the target).
In the present case, the answer is clearly the target language, which is further underscored by our restrictions on observability and compositionality; i.e.\@ that the parameters should not be observable in the source term; and that, for each recursive call to the translation function, the parameters should be derivable from the initial set.
Furthermore, we have added the criterion of parameter independence.
We believe that such a criterion will generally be necessary for encodings that allow the set of parameters to `evolve' or be updated in some structured way during the course of the translation, which seems particularly likely when we are working with structured names or terms. 
More recently, van Glabbeek \cite{glabbeek2018encodingexpressiveness} has also proposed a definition of a valid encoding, which he derives from a notion of a semantic equivalence or preorder, rather than basing it on a list of commonly agreed-upon criteria (as we have done in the present paper, following Gorla).
However, this work also does not consider parametrised translations.

Also related is the work by Carbone and Maffeis \cite{CARBONEMAFFEIS} on expressivity of polyadic synchronisation.
Their \EPI-calculus substitutes names for names (as in the \PI-calculus), but allows $n$-ary vectors of names \mbox{$x_1 \cdot \ldots \cdot x_n$} of arbitrary length $n \geq 0$ to appear in \emph{subject} position of input/output prefixes, and subjects must then match on all $n$ names to yield a reduction.
Thus name vectors can be altered at runtime, but they cannot grow in length as in the \RHO-calculus.
However, we could conceive of a (purposefully ill-sorted) variant of \EPI{} that would allow entire vectors of names to be substituted for single names, thereby allowing new vectors of increasing length to be composed at runtime.
We do not know if such a calculus could encode the \RHO-calculus, but we suspect that it might, if equipped with an appropriate notion of name equivalence.

Another approach to using structured terms as names is given by Bengtson et al.\@ \cite{bengtson2009psi, bengtson2011psi} and Parrow et al.\@ \cite{parrow2014higher} in their work on \PSI-calculi, which is based on the theory of nominal sets and datatypes by Gabbay and Pitts \cite{gabbay2002nominal}.
\PSI-calculi allow both subjects and objects to be terms from an arbitrary nominal datatype, and with substitution of terms for names. 
This enables runtime composition of terms, and, notably, the \RHO-calculus \emph{can} be instantiated as a (higher-order) \PSI-calculus, as the present author and others have shown in \cite{huttel2022hopsitypes}.

\section{Conclusion}
The original \RHO-calculus paper \cite{RHOCALC} by Meredith and Radestock raises some interesting questions about the nature of names in process calculi.
By including name generation in the language, it forces any process to give an explicit account of the source of any fresh names required during its execution, whilst this is entirely implicit in the \PI-calculus with the $\NEW{x}P$ operator.
This adds a degree of realism to the \RHO-calculus, which may be relevant from an implementation perspective, but also requires some extra care when we wish to reason about it formally.
For example, Meredith and Radestock attempted to show that the \PI-calculus can be interpreted in the \RHO-calculus, but their encoding did not properly account for the invariant that must hold for the names used as parameters in their encoding; i.e.\@ that the parameters always refer to the most recently replicated names, leading to two errors that invalidate their correctness result.
The purpose of the present paper has been to describe these errors and then give a new encoding of the \PI-calculus, for which we have shown correctness w.r.t.\@ a set of criteria for encodability close to those proposed by Gorla \cite{GORLA}.
The main difference is that we here use a parametrised translation, and we therefore had to take parameters into account in our criteria.
This seems unavoidable when we are working with a calculus with structured names like the \RHO-calculus, where all names are global and cannot be declared at runtime.

Our encoding works, modulo the criteria in Definition~\ref{def:valid_encoding}; yet it may not be an entirely satisfactory solution in at least one regard: the name server acts as a single, central source of fresh names.
If we consider the scenarios one might wish to \emph{model} in the \PI-calculus, having such a single central process might be acceptable for e.g.\@ models of programs running on a single computer, or models of client-server systems with a star topology.
However, for distributed systems with a different network topology, the translation would not yield an adequate representation.
Thus, the encoding may preserve the semantics of a program, but not necessarily the intuitions underlying its structure.
We could instead conceive of a more elaborate encoding, where e.g.\@ each replication also instantiates its own copy of a name server to service the replicated processes.
This would be closer to the intention in the encoding by Meredith and Radestock; but as we have seen, one would then have to be careful to ensure that each replica of the name server will generate a distinct namespace to avoid the possibility of a name clash.
This could be achieved by letting each replica first request fresh names for all its parameters, including the namespace root $s$ which must then be composed or otherwise shifted into a new namespace.
Yet this creates a scaffolding problem, where, in order to instantiate a new source of fresh names, one must first have a source of fresh names.
It does not remove the need for an initial, `top level' instance of the name server.
These considerations illustrate some of the difficulties involved in working with, and reasoning about, structured names with global visibility.
None of these problems are present in the \PI-calculus, yet any implementation of a \PI-calculus program would need to include a solution to the problem of obtaining fresh names.
In the words of Meredith and Radestock \cite{RHOCALC}, the \PI-calculus does not provide a `theory of names.'

We have also shown that the \PI-calculus cannot encode the \RHO-calculus in a way that satisfies the same criteria, modulo some requirements on the notion of behavioural equivalence $\BEHEQ$ used in Definition~\ref{def:beheq_requiements}.
The key to this separation result seems precisely to be the ability of the \RHO-calculus to create new \emph{free} names at runtime, which cannot be mimicked in the \PI-calculus.
This ability is a consequence of \emph{reflection} in the \RHO-calculus, which also gives it higher-order characteristics as a by-product.
In a process-calculus setting where computation is modelled as communication, higher-order behaviour appears as just a special case of reflection, where processes (code) are transmitted without modification. 
Thus, the separation result is also interesting in light of a remark by Sangiorgi regarding the encodability of \HOPI{} into the \PI-calculus.
He notes that this \emph{``[\ldots] proves that the first-order paradigm, being by far simpler, should be taken as basic. Such a conclusion takes away the interest in the opposite direction, namely the representability of the \PI-calculus within a language using purely communications of agents \ldots''} \cite[p.\@ 8]{SANGIORGIPHD}.
But as we have seen, this does not seem to hold in the more general case where higher-order characteristics derive from the capability of reflection.
The \RHO-calculus purely uses communication of agents (processes), because names and processes are the same thing.

\section*{Acknowledgements}
The author wishes to thank Hans Hüttel, Bjarke B. Bojesen and Alex R. Bendixen for many discussions of the \RHO-calculus, and Luca Aceto and the anonymous reviewers for their numerous and invaluable comments on earlier drafts of this paper.

\bibliographystyle{eptcs}
\bibliography{literature}
\end{document}
